\documentclass[conference,10pt]{IEEEtran}
\usepackage{amsmath, mathrsfs, mathtools}
\usepackage{amsfonts}
\usepackage{amssymb}
\usepackage{color}
\usepackage{multirow}
\usepackage{stmaryrd}
\usepackage{yfonts}
\usepackage{mathabx}
\allowdisplaybreaks
\usepackage{caption}
\usepackage{subcaption}
\usepackage{float}
\usepackage{stfloats}
\usepackage{amsfonts}
\usepackage{amssymb}
\usepackage{color}
\usepackage{multirow}
\usepackage{graphicx}
\usepackage{tcolorbox}
\usepackage{mathrsfs}
\usepackage[numbers,sort&compress]{natbib}
\usepackage{multicol}
\usepackage{algorithm}
\usepackage{algpseudocode}
\usepackage{pifont}
\usepackage{varwidth}
\usepackage{float}
\usepackage{afterpage}
\usepackage{balance}
\allowdisplaybreaks
\usepackage{lipsum}
\usepackage{mathbbol}

\def\mindex#1{\index{#1}}

\def\sq{\hbox{\rlap{$\sqcap$}$\sqcup$}}
\def\qed{\ifmmode\sq\else{\unskip\nobreak\hfil
\penalty50\hskip1em\null\nobreak\hfil\sq
\parfillskip=0pt\finalhyphendemerits=0\endgraf}\fi\medskip}

\long\def\defbox#1{\framebox[.9\hsize][c]{\parbox{.85\hsize}{%
\parindent=0pt
\baselineskip=12pt plus .1pt      
\parskip=6pt plus 1.5pt minus 1pt 
 #1}}}

\long\def\beginbox#1\endbox{\subsection*{}%
\hbox{\hspace{.05\hsize}\defbox{\medskip#1\bigskip}}%
\subsection*{}}

\def\endbox{}

\newsavebox{\junk}
\savebox{\junk}[1.6mm]{\hbox{$|\!|\!|$}}

\def\bC{{\mathbb C}}

\def\bE{{\mathbb E}}

\def\bR{{\mathbb R}}

\def\sfH{{\sf H}}

\def\bfmath#1{{\mathchoice{\mbox{\boldmath$#1$}}%
{\mbox{\boldmath$#1$}}%
{\mbox{\boldmath$\scriptstyle#1$}}%
{\mbox{\boldmath$\scriptscriptstyle#1$}}}}

\def\bfmY{\bfmath{Y}}

\def\bfmhhaY{\bfmath{\hhaY}} 
\def\bfmhhaY{\hbox to 0pt{$\widehat{\bfmY}$\hss}\widehat{\phantom{\raise 1.25pt\hbox{$\bfmY$}}}}

\def\til={{\widetilde =}}

\def\clC{{\cal C}}

\def\clN{{\cal N}}

 \def\FRAC#1#2#3{\genfrac{}{}{}{#1}{#2}{#3}}

\def\ddtp{{\mathchoice{\FRAC{1}{d^{\hbox to 2pt{\rm\tiny +\hss}}}{dt}}%
{\FRAC{1}{d^{\hbox to 2pt{\rm\tiny +\hss}}}{dt}}%
{\FRAC{3}{d^{\hbox to 2pt{\rm\tiny +\hss}}}{dt}}%
{\FRAC{3}{d^{\hbox to 2pt{\rm\tiny +\hss}}}{dt}}}}

\def\average#1,#2,{{1\over #2} \sum_{#1}^{#2}}

\def\eye(#1){{\bf(#1)}\quad}

\newtheorem{theorem}{{\bf Theorem}}

\newtheorem{lemma}{{\bf Lemma}}

\def\eq#1/{(\ref{e:#1})}

\newcommand{\beqn}[1]{\notes{#1}%
\begin{eqnarray} \elabel{#1}}

\newcommand{\eeqn}{\end{eqnarray} }

\newcommand{\beq}[1]{\notes{#1}%
\begin{equation}\elabel{#1}}

\newcommand{\eeq}{\end{equation}}

\def\bdes{\begin{description}}
\def\edes{\end{description}}

\newcounter{rmnum}

\newcounter{anum}

{\end{list}}

\def\ass(#1:#2){(#1\ref{#1:#2})}

\def\ritem#1{
\item[{\sf \ass(\current_model:#1)}]
}

\newenvironment{recall-ass}[1]{%
\begin{description}
\def\current_model{#1}}{
\end{description}
}

\newcommand{\Sigmay}{{\Sigmay}_{\yv}}
\def\wh{\widehat}

\def\herm{{\sfH}}

\newcommand{\snrul}{{\sf snr}_{\rm ul}}
\newcommand{\snrdl}{{\sf snr}_{\rm dl}}
\newcommand{\hvt}{\widetilde{\hv}}
\newcommand{\yvt}{\widetilde{\yv}}

\allowdisplaybreaks

\def\cg{{\clC\clN}} 
\newcommand{\normd}[1]{{\left\vert\kern-0.25ex\left\vert\kern-0.25ex\left\vert #1 
		\right\vert\kern-0.25ex\right\vert\kern-0.25ex\right\vert}}
\long\def\comment#1{}

\newcommand{\gv}{{\bf g}}
\newcommand{\hv}{{\bf h}}

\newcommand{\uv}{{\bf u}}
\newcommand{\wv}{{\bf w}}
\newcommand{\vv}{{\bf v}}
\newcommand{\xv}{{\bf x}}
\newcommand{\yv}{{\bf y}}
\newcommand{\zv}{{\bf z}}

\newcommand{\Bm}{{\bf B}}

\newcommand{\Gm}{{\bf G}}
\newcommand{\Hm}{{\bf H}}

\newcommand{\Sm}{{\bf S}}

\newcommand{\Um}{{\bf U}}

\newcommand{\Xm}{{\bf X}}

\newcommand{\Nc}{{\cal N}}

\newcommand{\Rc}{{\cal R}}

\newcommand{\Xc}{{\cal X}}

\newcommand{\lambdav}{\hbox{\boldmath$\lambda$}}

\newcommand{\phiv}{\hbox{\boldmath$\phi$}}
\newcommand{\psiv}{\hbox{\boldmath$\psi$}}

\newcommand{\Lambdam}{\hbox{\boldmath$\Lambda$}}

\newcommand{\Sigmam}{\hbox{\boldmath$\Sigma$}}
\newcommand{\Phim}{\hbox{\boldmath$\Phi$}}

\newcommand{\Psim}{\hbox{\boldmath$\Psi$}}

\newcommand{\trace}{{\rm Tr}}

\newcommand{\transp}{{\sf T}}

\newcommand{\hb}{\mathbb{h}}
\newcommand{\yb}{\mathbb{y}}

\title{Channel State Acquisition in FDD Massive MIMO: Rate-Distortion Bound and \\Effectiveness of ``Analog’’ Feedback}
\author{ Mahdi Barzegar Khalilsarai, Yi Song, Tianyu Yang, and Giuseppe Caire\\
	\vspace{2mm}
	$\{$m.barzegarkhalilsarai, yi.song, tianyu.yang, caire$\}$@tu-berlin.de
	\vspace{-3mm}
}

\begin{document}
	
	\maketitle
	
	\begin{abstract}
		We consider the problem of estimating channel fading coefficients (modeled as a correlated Gaussian vector) via Downlink (DL) training and Uplink (UL) feedback in wideband FDD massive MIMO systems. Using rate-distortion theory, we derive optimal bounds on the achievable channel state estimation error in terms of the number of training pilots in DL ($\beta_{\rm tr}$) and feedback dimension in UL ($\beta_{\rm fb}$) ,  with random, spatially isotropic pilots. It is shown that when the number of training pilots exceeds the channel covariance rank ($r$), the optimal rate-distortion feedback strategy achieves an estimation error decay of $\Theta (\text{SNR}^{-\alpha})$ in estimating the channel state, where $\alpha = \min (\beta_{\rm fb}/r , 1)$ is the so-called quality scaling exponent. We also discuss an ``analog" feedback strategy, showing that it can achieve the optimal quality scaling exponent for a wide range of training and feedback dimensions with no channel covariance knowledge and simple signal processing at the user side. Our findings are supported by numerical simulations comparing various strategies in terms of channel state mean squared error and achievable ergodic sum-rate in DL with zero-forcing precoding.
	\end{abstract}
	
	\section{Introduction}
	Massive MIMO consists of employing a large number of antennas ($M$) at the Base Station (BS) to simultaneously multiplex data over the spatial domain and serve a much smaller number of User Equipment (UEs) ($K$) at the same time-frequency resource in the Downlink (DL) \cite{larsson2014massive,boccardi2014five}. Achieving the capacity improvements of massive MIMO requires availability of accurate channel state information at the multi-antenna BS transmitter (CSIT). Therefore the BS needs to learn (or \textit{train}) the CSIT, namely the fading coefficients associated with each of its $M$ antennas and those of the $K$ UEs. In the time division duplexing (TDD) operation mode, the BS learns the CSIT by receiving pilots from the UEs in Uplink (UL) and, relying on UL-DL channel reciprocity, extrapolating the DL channel \cite{marzetta2006much}. In frequency division duplexing (FDD) mode however, channel reciprocity does not hold and the BS needs to broadcast training pilots in the DL to the UEs and receive their estimated channel state via explicit feedback in UL. The error in estimating the CSIT strongly affects the DL spectral efficiency. For example, it is well-known that when the error between true and estimated CSIT decreases as $O(\text{SNR}^{-\alpha})$ in SNR (equivalently, BS transmission power) for some constant $\alpha \in [0,1]$, then zero-forcing (ZF) precoding can achieve only a fraction $\alpha$ of the optimal degrees of freedom (DoF) per UE \cite{caire2007multiuser,jindal2006mimo}. Therefore the choice of CSIT training and feedback strategies is crucial in achieving faster estimation error decays. 
	
	In this paper, we study the CSIT estimation problem for wideband massive MIMO systems, in which the channel is modeled as a spatially correlated, stationary Gaussian random process that evolves in time according to a block-fading model \cite{biglieri1998fading}. We consider a generic design of random, spatially isotropic pilots of arbitrary dimension $\beta_{\rm tr}$ and we derive a lower bound on the achievable CSIT estimation error using rate-distortion theory and the idea of \textit{remote source coding} at the UE. This lower-bound results in an upper-bound on the achievable quality scaling exponent as a function of training and feedback dimension pairs $(\beta_{\rm tr},\beta_{\rm fb})$, showing the fastest rate of error decay among all feedback strategies. If covariance knowledge is available at the UE, we demonstrate how one can approach the optimal performance at the price of small overhead in the number of feedback bits using the entropy-coded scalar quantization (ECSQ) \cite{ziv1985universal}. We then study a variation of the well-known analog feedback (AF) strategy \cite{samardzija2006unquantized,thomas2005obtaining,marzetta2006fast}, in which the training measurements are sent to the BS via unquantized quadrature amplitude modulation (QAM), using which the BS computes an MMSE estimate of the channel (channel covariance knowledge at the BS is assumed). This variation of AF is an attractive strategy, because it requires no channel covariance knowledge and no sophisticated processing (namely, channel estimation and quantization) at the UE side, both of which come at a high price in the case of wideband massive MIMO channels. We emphasize this point by deriving an expression for the achievable quality scaling exponent with AF, and showing its optimality for a wide range of choices of training and feedback dimensions (see Fig. \ref{fig:quality_scaling_exponent}). 
	
	The effect of channel training and feedback on CSIT estimation error and spectral efficiency in FDD MIMO systems has been studied before in several works \cite{kobayashi2011training,caire2007multiuser,caire2007quantized,jiang2015achievable,shirani2009channel}. In \cite{caire2007multiuser} and \cite{caire2007quantized} lower bounds on the achievable DL rate with ZF precoding for analog and digital (quantized) feedback are given, where it is assumed that the number of training pilots exceeds the channel dimension. In \cite{jiang2015achievable} achievable rates of an FDD massive MIMO system with optimized training pilots and with channel covariance knowledge at the UE side was studied. In \cite{bazzi2018amount} the authors studied sufficient conditions to achieve full DoF by considering channel covariance knowledge at the UE and an error-free channel state feedback to the BS (i.e. an ideal feedback link). The present work provides the following novelties with respect to the above:
	\begin{enumerate}
		\item We consider training and feedback of the channel on the whole bandwidth of $N$ subcarrier, while most previous works assume a narrow-band model, neglecting frequency-domain channel correlation. From a practical standpoint, the recent releases of the 5G new radio heavily emphasize on exploiting this frequency-domain correlation feature of the channel to reduce feedback overhead \cite{gpp2020phylayer}. This aspect is captured in our model and reflected in the rate-distortion bound as well as the proposed AF strategy.
		\item We consider training the channel with an arbitrary number of pilots ($\beta_{\rm tr}$) that can be larger or smaller than the CSI dimension ($ MN$). Most previous works have assumed training with a number of pilots \textit{larger} than the channel dimension, which is impractical in massive MIMO where the CSI dimension potentially exceeds the dimension of the coherence block.
		\item We provide optimal information-theoretic bounds on the achievable CSIT estimation error and the quality scaling exponent in terms of training and feedback dimensions. To the best of our knowledge, such analysis for massive MIMO channels with spatial correlation has not been considered before.
		\item We show the effectiveness of AF without the need for either channel covariance knowledge, or sophisticated estimation and quantization at the UE side.
	\end{enumerate}

	The rest of the paper is organized as follows. In Section \ref{sec:ch_training} we describe the model, training method and error metrics. In Section \ref{sec:RD_LB} we derive a lower bound on CSIT estimation error via rate-distortion theory and we explain feedback via ECSQ. In Section \ref{sec:analog_feedback} we discuss AF and derive an expression for the quality scaling exponent it can achieve. Finally, Section \ref{sec:numerical_results} concludes the paper with numerical simulation results.

	\textbf{Notation:} We denote scalars, vectors, and matrices by small, small bold-face and capital bold-face letters $x,\, \xv, \Xm$, respectively. For a positive integer $N$, we define $[N]\triangleq \{ 1,\ldots,N\}$. Superscripts $(\cdot)^\transp$ and $(\cdot)^\herm$ denote transpose and Hermitian transpose, respectively. The function $\mathbf{1}_{\Xc} (x)$ is the indicator function such that $\mathbf{1}_{\Xc} (x)=1$ if $x\in \Xc$ and $\mathbf{1}_{\Xc}(x)=0$ if $x\notin \Xc$. For a real-valued functions $f(\cdot)$ and $g(\cdot)$, defined over all positive reals $x>0$, we say $f(x)=\Theta (g(x))$ if $\exists A_1,A_2>0$ such that $A_1 g(x) \le f(x)\le A_2 g(x)$ for all $x\ge x_0$.
	
	\section{Channel Training}\label{sec:ch_training}
	We consider a broadcast MIMO OFDM system, consisting of a BS with an array of $M$ antennas and $K\le M$ single-antenna user equipment (UEs). The frequency-domain signal corresponding to subcarrier $n$ received by an arbitrary UE can be expressed as $y[n] =\hvt^\herm [n]  \xv [n] + z[n],$ where $\hvt [n]  \in \bC^M$ contains the channel fading coefficients between the BS and the UE at subcarrier $n$, $\xv[n]  \in \bC^M$ is the transmit signal satisfying the power constraint $\bE [\Vert \xv [n]\Vert^2]\le \snrdl $ for all $n$, where $\snrdl$ denotes SNR in the DL, while $z\sim \cg (0,1)$ is additive white Gaussian noise (AWGN).\footnote{This definition of transmit power and noise variance simplifies notation, since we only need the ratio of the two, i.e. the SNR.} We concatenate the channel over all subcarriers in a vector $\hv = [ \hvt[1]^\transp,\ldots, \hvt[N]^\transp ]^\transp \in \bC^{MN}$, which we refer to as the channel state information at the transmitter (CSIT). We assume that $\hv$ evolves according to a block-fading model in which it is constant over frames of length $T$ and changes from frame to frame according to an ergodic stationary, spatially correlated, zero-mean Gaussian process, i.e. $\hv\sim \cg (\mathbf{0},\Sigmam^h)$ where $\Sigmam^{h} = \bE [\hv \hv^{\herm}]$ is the channel covariance of rank $r = {\rm rank}(\Sigmam^h)$. 
	
	Throughout this work we assume that UEs have perfect estimates of their channels (perfect CSIR). We then focus on CSIT acquisition by the BS via the following process. Given a set $\Nc_p\subseteq [N]$ of $N_p$ pilot subcarriers, the BS broadcasts a sequence of $T_p \le T$ training vectors per pilot subcarrier. The training measurements received at the UE can be written as
	\begin{equation}\label{eq:subcar_measurements}
		\yvt^{\rm tr}[n] = \hvt^\herm [n] \widetilde{\Xm}^{\rm tr}  [n]+\widetilde{\zv}^{\rm tr} [n],~n\in \Nc_p,
	\end{equation}
	where $\widetilde{\Xm}^{\rm tr}[n] = \left[\widetilde{\xv}_1[n],\ldots,\widetilde{\xv}_{T_p}[n]\right]\in \bC^{M\times T_p}$ is a matrix containing the $T_p$ pilot vectors as its columns. The training dimension $\beta_{\rm tr} = T_p N_p$ denotes the total number of dimensions dedicated to training in a time-frequency block of dimension $TN$. The $\beta_{\rm tr}$ training measurements at the UE can be represented by a single vector $\yv^{\rm tr} = [\yvt [n_1], \ldots,\yvt [n_{N_p}]]\in \bC^{1\times \beta_{\rm tr}}$, where using \eqref{eq:subcar_measurements} we have
	\begin{equation}\label{eq:training_signal}
		\yv^{\rm tr} =\hv^\herm  \Xm^{\rm tr} + \zv^{\rm tr},
	\end{equation}
	where $\Xm^{\rm tr} = \left( \Bm_{n,\ell} \right)_{n\in [N]}^{\ell\in [N_p]}\in \bC^{MN\times \beta_{\rm tr}}$ is the \textit{training matrix} consisting of $M\times T_p$ blocks $\Bm_{n,\ell}$, where $\Bm_{n,\ell} = \widetilde{\Xm}^{\rm tr}[n_{\ell}]$ if $n_\ell \in \Nc_p$ and $\Bm_{n,\ell} = \mathbf{0}$ if $n_\ell \notin \Nc_p$.  The training matrix can be designed in several ways, for example by optimizing various performance criteria based on channel covariance knowledge \cite{kotecha2004transmit,jiang2015achievable}. However, we consider a simpler, and therefore practically more available design, in which pilot vectors are randomly and independently generated according to an isotropic Gaussian distribution,
	\begin{equation}\label{eq:isotropic_pilots}
		\widetilde{\xv}_i^{\rm tr}[n] \sim \cg (\mathbf{0},\frac{\snrdl}{M}\mathbf{I}),~i\in [T_p], \, n\in \Nc_p,
	\end{equation}
	where we verify that this design satisfies the transmission power constraint $\bE[ \Vert \widetilde{\xv}_i^{\rm tr}[n] \Vert^2]\le \snrdl$. In other words, the elements of $\Xm^{\rm tr}$ that are not identically zero, are generated as $\cg (0,\snrdl/M)$ Gaussian random variables.

	After receiving pilot symbols, the UE computes a message containing information about the channel state and sends it to the BS via $\beta_{\rm fb}$ uses of the UL channel. Given the feedback channel output, the BS computes an estimate $\widehat{\hv}$ of the CSIT. We consider the mean squared error (MSE)
	\begin{equation}\label{eq:err_def}
		d(\hv,\widehat{\hv}) = \bE \left[ \Vert \hv - \widehat{\hv}\Vert^2 \right]
	\end{equation}
	as the error (distortion) metric between true and estimated CSIT. For a fixed tuple $(\beta_{\rm tr},\beta_{\rm fb},\snrdl)$ and a given realization of $\Xm^{\rm tr}$, we say that an error $D$ is achievable if $d(\hv,\widehat{\hv})\le D$. Accordingly, we say that a quality scaling exponent of $\alpha$ is achievable if $d(\hv,\widehat{\hv}) = \Theta (\snrdl^{-\alpha})$, with the exponent indicating how fast the error decays with SNR \cite{joudeh2016sum,jindal2006mimo}. Note that an error decay of $ \Theta (\snrdl^{-\alpha})$ in estimating the wideband channel implies an error decay in estimating the channel over each subcarrier that is at least as fast. In other words, if $d(\hv,\widehat{\hv}) = \Theta (\snrdl^{-\alpha})$, then $d(\widetilde{\hv}[n],\widehat{\widetilde{\hv}}[n])  = O (\snrdl^{-\alpha})$, where $\widehat{\widetilde{\hv}}[n]$ is the estimate of $\widetilde{\hv}[n]$ (the channel over subcarrier $n$). 
	
	The quality scaling exponent is related to the system DoF as follows. It is known that for a multi-user system with $K\le M$ UEs, if for some $\alpha \in [0,1]$ the MSE in estimating the CSIT decays as $O(\snrdl^{-\alpha})$, then ZF precoding achieves a total DoF of $K\alpha$, with $\alpha = 1$ corresponding to the full DoF \cite{davoodi2016aligned}. One can achieve a slightly improved DoF of $1+(K-1)\alpha$ with rate-splitting \cite{joudeh2016sum}, which is coincidentally also an upper-bound, i.e. no scheme can achieve a higher DoF. It is therefore apparent that, when $\alpha=0$, we can achieve a maximum DoF of 1 with rate-splitting, whereas $\alpha = 1$ yields a full DoF of $K$.
	
	\subsection*{The Feedback Channel}In what follows we model the UL as a MIMO-MAC channel where all the UEs send their feedback simultaneously to the BS. We assume for simplicity that the BS has perfect knowledge of the UL channel and we use the high-SNR capacity formula $C^{\rm ul} = \log (1+M\snrul)$ with a so-called diversity-multiplexing trade-off  factor of one \cite{caire2007multiuser}. We assume the SNR in UL to be proportional to the DL SNR as $\snrul = \kappa\, \snrdl, $
	with $\kappa>0$ being a positive constant. The modeling of the feedback link as such was considered in \cite{jiang2015achievable} and is a simplifying assumption that allows for a meaningful and elegant development of the theory but is not fundamental, in the sense that one can obtain similar results by considering other feedback channel models. 
	
	\section{Lower bound on CSIT estimation error via rate-distortion theory}\label{sec:RD_LB}
	To derive a lower-bound on the MSE in estimating the channel for given training and feedback dimensions, we think of the UE as an \textit{encoder} that aims at encoding a \textit{source} that produces channel realizations $\hv$, given the training noisy linear measurements of those realizations as in \eqref{eq:training_signal}. Since the encoder does not have direct access to the source output, this problem is an example of \textit{remote source coding} \cite{berger1971rate}. Following standard information theoretic notation, we can formulate this problem by modeling the source as a stationary sequence of vector symbols $\hb^{\Delta} = \{ \hv^{(i)} \}_{i=1}^{\Delta} $ with distribution $\hv^{(i)}\sim \cg (\mathbf{0},\Sigmam^h)$, where $\hv^{(i)}$ denotes the channel in frame $i$ and $\Delta$ is the total number of frames. The UE observes a sequence of measurements of the source as $\yb^{{\rm tr}\, \Delta}=\{ \yv^{{\rm tr}\, (i)} \}_{i=1}^{\Delta}$, where from \eqref{eq:training_signal} we have $\yv^{{\rm tr}\, (i)} = \hv^{(i)\, \herm}\Xm^{\rm tr} + \zv^{{\rm tr}\, (i)}$. Because we seek a lower-bound on the achievable error, we can assume that the UE knows the channel statistics and encodes the source given measurements over infinitely many blocks. Now, a $(2^{\Delta R},\Delta)$ remote rate-distortion code consists of a sequence of encoding $ f_{\Delta}\,: \yb^{{\rm tr}\, \Delta}\to [2^{\Delta R}],$
	and decoding $  g_{\Delta}\,: [2^{\Delta R}]\to \widehat{\hb}^\Delta $ functions, where $R\ge 0$ denotes the code rate and where $\widehat{\hb}^{\Delta} = \{ \widehat{\hv}^{(i)}\}_{i=1}^{\Delta}$ is the sequence of channel estimates. A remote rate-distortion pair $(R,D)$ is said to be achievable if there exists a sequence of $(2^{\Delta R},\Delta)$ codes such that $\underset{\Delta \to \infty}{\lim}  \frac{1}{\Delta} \sum_{i=1}^{\Delta} d(\hv^{(i)},\widehat{\hv}^{(i)})\le D$, where $d(\cdot,\cdot)$ is defined in \eqref{eq:err_def}. The closure of all such pairs is the rate-distortion region.
	
	\textbf{Remark:} Note that here the encoder is only required to yield an average error less than $D$ over all frames. This is a weaker condition in comparison to the per-frame achievable error defined in the previous section as the condition $d(\hv,\widehat{\hv})\le D$, in the sense that, if a feedback strategy achieves the latter, it also achieves the former. It follows that, for the same training and feedback rates, the achievable error with these assumptions serves as a strict lower-bound for the achievable error of all feedback strategies that operate over finite blocks of the source, in which the UE has no access to channel statistics, and consider the stronger notion of per-frame achievable error as defined in the previous section. 
	
	 The remote rate-distortion function $R_{\hv}^r(D)$ is the infimum of rates $R$ such that $(R,D)$ is in the rate-distortion region for given $D$, and shows the infimum number of bits needed to quantize a single channel vector to achieve an error of $D$. We derive an expression for the remote rate-distortion function via the following lemma.
	\begin{lemma}[remote rate-distortion function]\label{lem:rate_distortion}
		For a fixed realization of the training matrix $\Xm^{\rm tr}$, let $\Sigmam^u = \Sigmam^h \Xm^{\rm tr}\left(  \Xm^{{\rm tr}\, \herm} \Sigmam^h \Xm^{{\rm tr}} + \mathbf{I} \right)^{-1}\Xm^{{\rm tr}\, \herm} \Sigmam^h$ denote the covariance of the posterior mean of the channel given pilot measurements $\uv^{(i)} = \bE[\hv^{(i)}|\yv^{{\rm tr}\, (i)} ]$ and denote its eigenvalues by $\{\lambda_{\ell}^{u}\}_{\ell=1}^{MN}$. The remote rate-distortion function is given by 
		\begin{equation}\label{eq:rrd_expression}
			R_{\hv}^r (D) = \sum_{\ell=1}^{MN}\left[ \log \frac{\lambda_{\ell}^{u}}{\gamma} \right]_+,~\text{for}~ D\ge D_{\rm mmse}
		\end{equation}
		where $\gamma$ is chosen such that $\sum_{\ell=1}^{MN}\min \{ \gamma, \lambda_{\ell}^{u} \} = D-D_{\rm mmse},$ and where $D_{\rm mmse} = d( \hv^{(i)} , \uv^{(i)} )$.
	\end{lemma}
	\begin{proof}
		See Appendix \ref{app:RD_lemma_proof}.
	\end{proof}
Accordingly, we can define the remote distortion-rate function $D_{\hv}^r (R)$ as the infimum of errors $D$ such that $(R,D)$ is in the rate-distortion region for given $R$ and is equivalent to the inverse of $R_{\hv}^r(\cdot)$ \cite{cover2006elements}. 

With rate-distortion feedback strategy, the UE remotely encodes the channel using $R_{\hv}^r(D)$ bits and sends the quantization index in the UL via a channel code. A direct application of the source-channel separation theorem with distortion (see \cite{cover2006elements}, exercise 10.17) yields that an end-to-end error of $D$ between the channel vector and its estimate at the BS via feedback over a channel of capacity $C^{\rm ul}$ is achievable if and only if $\beta_{\rm fb}C^{\rm ul} >R_{\hv}^r(D)$. Thereby for a given feedback dimension $\beta_{\rm fb}$, we can achieve an error of 
	\begin{equation}\label{eq:distor_bound}
		D > D_{\hv}^r ( \beta_{\rm fb} C^{\rm ul} )
	\end{equation}
	Therefore, for given $\beta_{\rm fb}$ and $\beta_{\rm tr}$, the lower-bound on the achievable error is $D_{\hv}^r ( \beta_{\rm fb} C^{\rm ul} )$, where the dependence on $\beta_{\rm tr}$ is implicit in the expression for the rate-distortion function $D_{\hv}^r$. It seems very difficult to make this relation more explicit in the general case, but we make it  explicit for the large SNR regime via the following theorem. 
\begin{theorem}\label{thm:rate_asymp}
	The rate-distortion feedback strategy achieves a CSIT estimation error of $\Theta (\snrdl^{-\alpha_{\rm rd}})$ with probability one over the realizations of the training matrix $\Xm^{\rm tr}$, where
	\begin{equation}\label{eq:decay_exponent_ub}
		\alpha_{\rm rd}= \min (\beta_{\rm fb}/r,1)\mathbf{1}_{ [r,\infty) }(\beta_{\rm tr}) .
	\end{equation}
	is the quality scaling exponent.
\end{theorem}
	\begin{proof}
		See Appendix \ref{app:rate_asymp_thm_proof}.
	\end{proof}
	Note that in this theorem $r = {\rm rank}(\Sigmam^h)$ is the channel covariance rank. The quality scaling exponent of the rate-distortion quantizer is an upper-bound on the quality scaling exponent of all feedback strategies and is illustrated as a heat map in Fig. \ref{fig:quality_scaling_exponent} (left).  The resulting system DoF with rate-splitting in this case is given by 
	\begin{equation}\label{eq:DoF_rd}
		\text{DoF}_{\rm rd}   = 1 + \mathbf{1}_{ [r,\infty) }(\beta_{\rm tr})\min (\beta_{\rm fb}/r,1)(K-1).
	\end{equation}

		\begin{figure*}[t]
		\centerline{\includegraphics[scale=0.6]{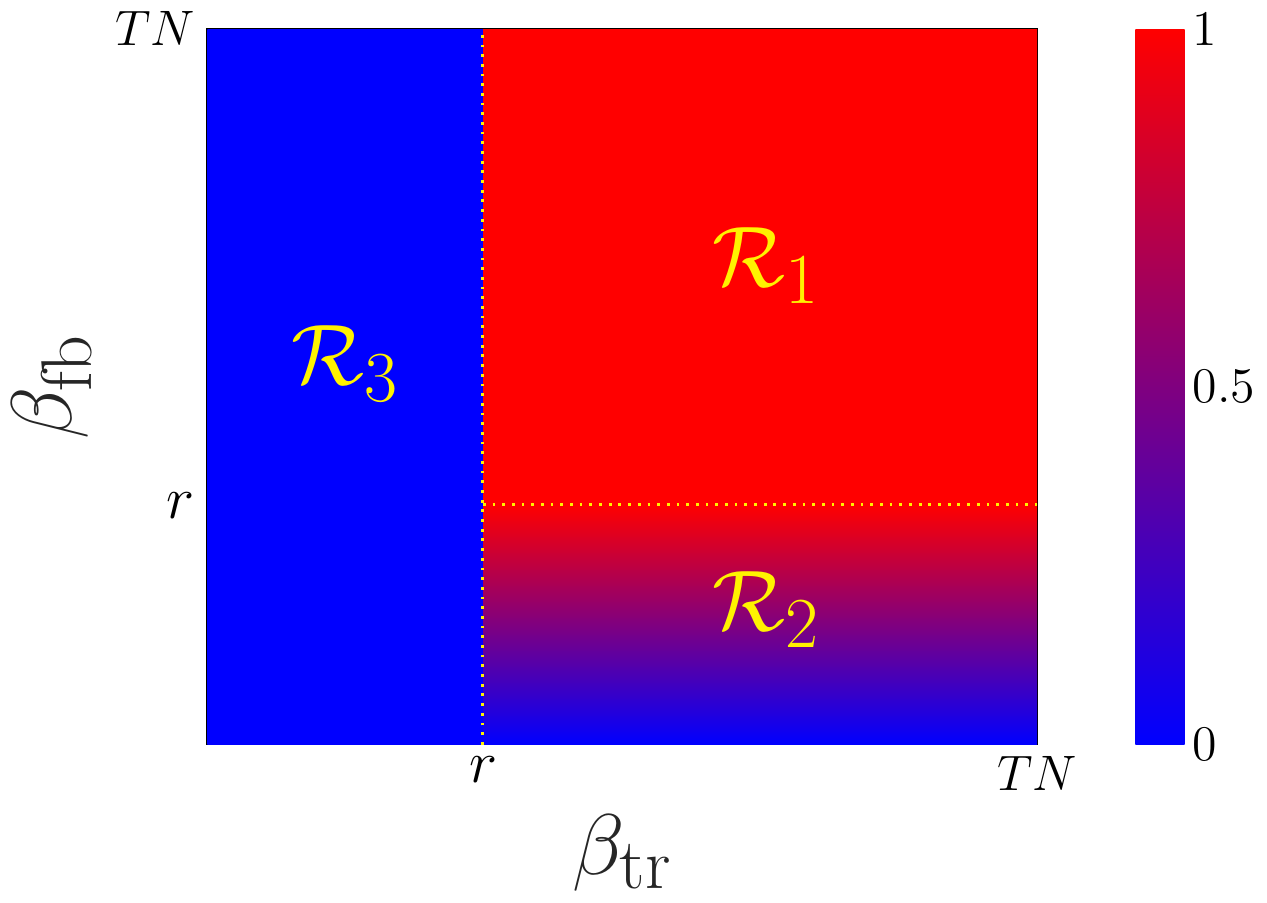}
			\includegraphics[scale=0.6]{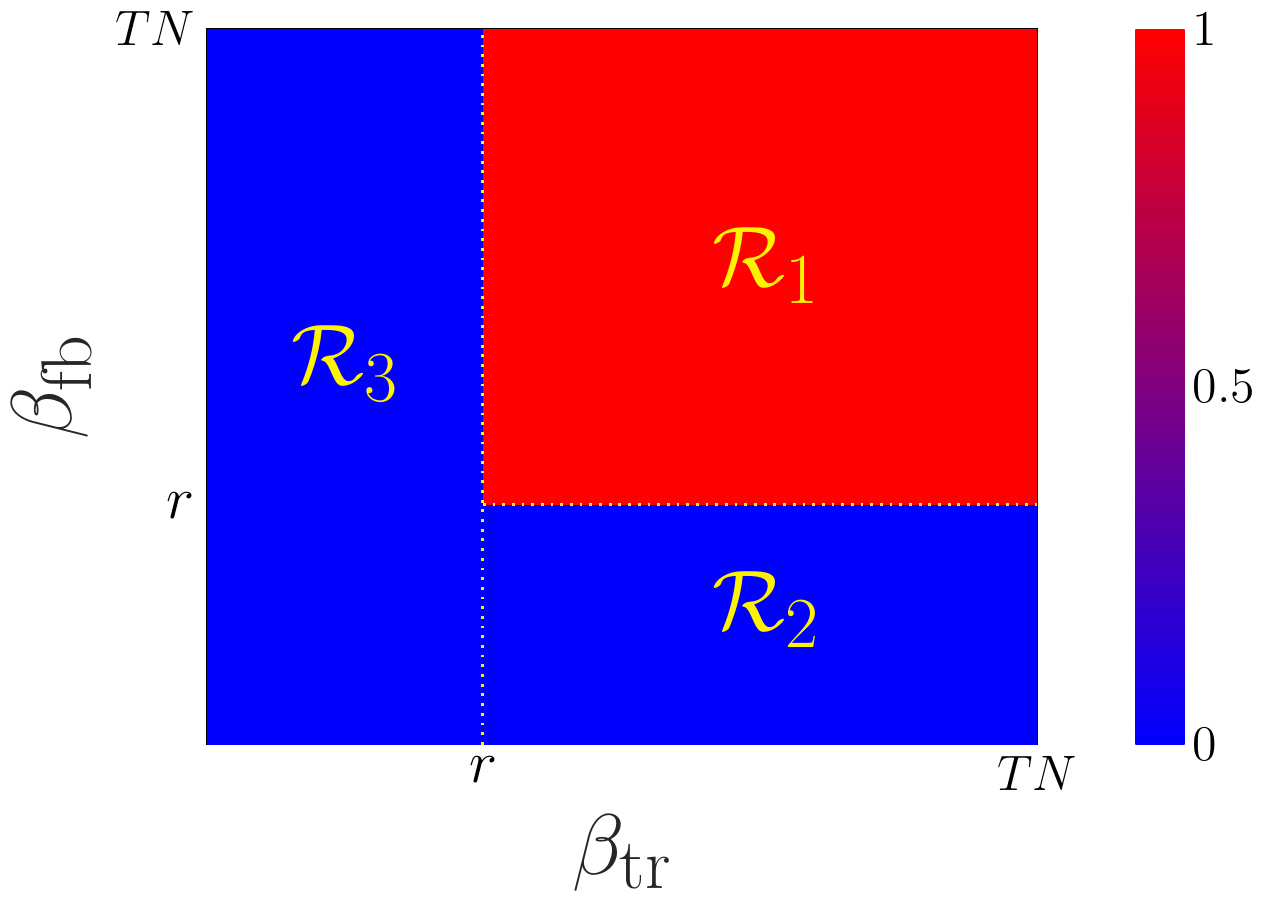}}
		\caption{The quality scaling exponent as a function of $(\beta_{\rm tr},\beta_{\rm fb})$, represented as a heat-map. The left figure corresponds to the optimal rate-distortion feedback $\alpha_{\rm rd}$, whereas the right figure corresponds to AF $\alpha_{\rm af}$. Three regions for training and feedback dimension parameters are distinguished by $\Rc_1$, $\Rc_2$ and $\Rc_3$. In regions $\Rc_1$ and $\Rc_3$ the quality scaling exponent of AF coincides with that of the rate-distortion feedback, whereas in region $\Rc_2$ it is strictly sub-optimal ($\alpha_{\rm af}=0$ vs $\alpha_{\rm rd} = \beta_{\rm fb}/r$).}
		\label{fig:quality_scaling_exponent}
	\end{figure*}
	\subsection*{Entropy-Coded Scalar Quantization and Feedback}
	The rate-distortion bound is achieved by vector quantization applied to a large block of MMSE channel estimates $\{\uv^{(i)} \}_{i=1}^{\Delta}$. This is impractical, given the high channel and block dimensions and the notorious difficulty of designing optimal vector quantizers. Therefore, given channel covariance knowledge, to produce the feedback message the UE can employ a much simpler \textit{entropy-coded scalar quantizer} (ECSQ) as follows. Given the training vector, the UE computes the MMSE channel estimate $\uv = \bE [\hv | \yv^{\rm tr}]$. The Karhunen-Lo{\`e}ve (KL) expansion of $\uv$ can be written as $\uv = \sum_{\ell=1}^{MN} w_\ell \gv_\ell$, where $\gv_\ell$ are eigenvectors of the covariance $\Sigmam^u$ and $w_{\ell} \sim \cg(0,\lambda_{ \ell}^u)$ are complex Gaussian coefficients with variance $\lambda_{ \ell}^u$, equivalent to the eigenvalues of $\Sigmam^u$. The idea is to quantize the vector of coefficients $\wv = [w_1,\ldots,w_{MN}]^\transp$, component by component, so as to achieve an error $d(\uv ,\widehat{\uv}) = \bE [\Vert \uv - \widehat{\uv} \Vert^2 ]  = \bE [\Vert \wv - \widehat{\wv} \Vert^2 ] \le D- D_{\rm mmse}$. Inspired by the \textit{reverse water filling} formulation used to derive the remote rate-distortion function in \eqref{eq:rrd_expression} (see proof of Lemma \ref{lem:rate_distortion}), the UE only quantizes those coefficients $w_{\ell}$ for which $\lambda_{ \ell}^u \ge \gamma$, where $\gamma$ is chosen such that $\sum_{\ell=1}^{MN}\min \{ \gamma, \lambda_{\ell}^{u} \} = D-D_{\rm mmse}$. With this choice, the error of quantizing each coefficient is given by $E[|w_{\ell}-\widehat{w}_{\ell}|^2]= \min(\lambda_{ \ell}^u,\gamma) $ and the associated rate is $ [\log\frac{\lambda_{ \ell}^u}{\gamma} ]_+$ bits. Now, instead of using the optimal rate-distortion vector quantizer, the UE can use a \textit{dithered} scalar quantizer to encode coefficients with variance above $\gamma$. Using a classic result from Ziv \cite{ziv1985universal}, one can show that this quantizer achieves an error $\gamma $ with a number of bits $ b_{\ell}=  \log\frac{\lambda_{ \ell}^u}{\gamma} + 1.508.$ No bits are assigned to coefficients for which $\lambda_{ \ell}^u<\gamma$. Since both the UE and the BS know the channel covariance, there is no need to encode the position of quantized coefficients. Therefore, the total number of feedback bits to achieve an error $D$ with ECSQ can be computed as 
	\begin{equation}
		R_{\rm scalar} (D) =\sum_{\ell : \lambda_{ \ell}^u\ge \gamma}\left( \log\frac{\lambda_{ \ell}^u}{\gamma} + 1.508\right).
	\end{equation}	
	ECSQ can be seen as a practical quantizer that can have a performance close to optimal (in the sense of achievable error), when the UE has access to the channel covariance knowledge and at the price of higher feedback rate.

	\section{Analog Feedback}\label{sec:analog_feedback}
	In analog feedback the UE extracts its $\beta_{\rm tr}$ received DL pilot symbols $\yv^{\rm tr}$ from the DL training and feeds them back to the BS via Quadrature Amplitude Modulation (QAM) symbols with unquantized I and Q components via $\beta_{\rm fb} = \zeta \beta_{\rm tr}$ channel symbols \cite{samardzija2006unquantized,thomas2005obtaining}. In particular, the training vector is modulated by a full-rank matrix $\Psim$ of dimension $\beta_{\rm tr} \times \beta_{\rm fb}$, known to both UE and BS. The received feedback at the BS is given by 
	\begin{equation}\label{eq:af_bs_signal}
		\yv^{\rm af} = \yv^{\rm tr} \Psim + \widetilde{\zv} = \hv^\herm \Xm^{\rm tr} \Psim +\zv^{\rm af},
	\end{equation}
	where $\widetilde{\zv}\sim \cg (\mathbf{0},\mathbf{I})$ is the AWGN over the feedback channel and $\zv^{\rm af} = \zv^{\rm tr}\Psim +\widetilde{\zv}$. The scalar $\zeta=\beta_{\rm fb}/\beta_{\rm tr}$ denotes the number of feedback channel uses per training coefficient. It is customary to choose $\zeta \ge 1$ (so that each training symbol is fed back at least once) and to select $\Psim$ to be a unitary ``spreading" matrix ($\Psim \Psim^\herm = \mathbf{I}$) \cite{marzetta2006fast,caire2007multiuser}. However, we allow $\zeta$ to be any positive value to keep the generality of the problem. Therefore $\Psim$ is only required to be full-rank (not necessarily unitary), and we nevertheless call it the spreading matrix for simplicity. Considering the MIMO-MAC capacity formula $C^{\rm ul} = \log (1+M\snrul)$ and from the feedback model \eqref{eq:af_bs_signal}, the feedback channel input has to satisfy the per-symbol average power constraint $\bE[|x|^2]\le M \snrul$. Therefore, the $i$-th column of the spreading matrix $\psiv_i$ is chosen such that
	\begin{equation}\label{eq:psi_vecs}
		\psiv_i^\herm \Sigmam_{\yv^{\rm tr}} \psiv_i = M \snrul, ~ i \in [\beta_{\rm fb}],
	\end{equation}
	where $\Sigmam_{\yv^{\rm tr}} = \bE [\yv^{\rm tr}\yv^{{\rm tr}\, \herm} ] = \Xm^{{\rm tr}\, \herm} \Sigmam^h \Xm^{\rm tr} + \mathbf{I}$ is the covariance of the training vector. Note that selecting a set of $\beta_{\rm fb}$ vectors that satisfy \eqref{eq:psi_vecs} and which contain a subset of $\min(\beta_{\rm tr},\beta_{\rm fb})$ linearly independent elements is always possible because $\Sigmam_{\yv^{\rm tr}}$ is of rank $\beta_{\rm tr}$.  
	The BS computes the minimum MSE (MMSE) estimate of the full-dimensional channel given the feedback as 
	\begin{equation}
		\wh{\hv} = \bE\left[\hv | \yv^{\rm af} \right] = \Sigmam^h \Xm^{\rm tr} \Psim \Sigmam_{\yv^{\rm af}}^{-1} \yv^{\rm af}, 
	\end{equation}
where $\Sigmam_{\yv^{\rm af}} = \Psim^\herm \Xm^{{\rm tr}\herm} \Sigmam^h \Xm^{\rm tr} \Psim +\Psim^\herm \Psim + \mathbf{I}  $. Note that unlike the rate-distortion quantizer and the ECSQ, with AF we do \textit{not} assume channel covariance knowledge at the UE. The CSIT estimation error with AF can be computed as
	\begin{equation}
		\begin{aligned}
			D &= \bE [ \Vert\hv - \widehat{\hv} \Vert^2 ]\\
			& = {\rm Tr}\left( \Sigmam^h -  \Sigmam^h \Xm^{\rm tr} \Psim \Sigmam_{\yv^{\rm af}}^{-1} \Psim^\herm \Xm^{{\rm tr}\, \herm} \Sigmam^h\right)\\
		\end{aligned}
	\end{equation}
	where ${\rm Tr}(\cdot)$ denotes the trace. The following theorem shows the scaling law of this error for large SNR.
	\begin{theorem}\label{thm:af_distortion_modified}
		The analog feedback strategy achieves a CSIT estimation error of $\Theta (\snrdl^{-\alpha_{\rm af}})$ with probability one over the realizations of the training matrix $\Xm^{\rm tr}$, where
		\begin{equation}\label{eq:af_decay_exponent}
			\alpha_{\rm af} = \mathbf{1}_{ [r,\infty) }\left(\min(\beta_{\rm tr},\beta_{\rm fb}) \right),
		\end{equation}
		is the quality scaling exponent.
	\end{theorem}

	\begin{proof}
		See Appendix \ref{app:af_thm_proof_modified}.
	\end{proof}
	The resulting system DoF with AF is given as
	\begin{equation}\label{eq:DoF_af}
		\text{DoF}_{\rm af}  = 1 + \mathbf{1}_{[r,\infty)}\left(\min(\beta_{\rm tr},\beta_{\rm fb})\right)(K-1).
	\end{equation}

	The quality scaling exponent of AF is illustrated as a heat map in Fig. \ref{fig:quality_scaling_exponent} (right), where a comparison between the right and left figures shows that the exponent achieved by AF is the same as that of the rate-distortion quantizer for all training and feedback dimensions $(\beta_{\rm tr},\beta_{\rm fb})$ belonging to regions $\Rc_1$ and $\Rc_3$. In region $\Rc_2$, rate-distortion feedback achieves an exponent of $\beta_{\rm fb}/r>0$, whereas AF has exponent zero, and is therefore strictly sub-optimal. 
	
	\section{Numerical Results}\label{sec:numerical_results}
	We consider a ULA with $M=32$ antennas at the BS communicating with $K=6$ UEs over a total of $N=24$ OFDM subcarriers. The channel coherence time is assumed to be $5$ ms, corresponding to $T=5\times 14 = 70$ OFDM symbols in LTE \cite{sesia2011lte}.\footnote{The coherence time can vary to emulate fast-varying (smaller $T$) and slow-varying (large $T$) channels.}  We consider $N_p = 4$ pilot subcarriers, uniformly placed one per $6$ subcarriers. Different training dimensions are considered by varying the number of pilots sent per pilot subcarrier, i.e. by changing the variable $T_p$. With these parameters, the isotropic pilot vectors are generated according to \eqref{eq:isotropic_pilots}. We produce UE channels according to a multipath model \cite{sayeed2003virtual,bajwa2008compressed}, as 
		\begin{equation}
			h_{m,n} = \sum_{\ell=1}^L c_{\ell} e^{j\pi m \frac{2 d}{\lambda}\sin \theta_{\ell} }	e^{-j2\pi n \Delta f  \tau_{\ell}},
		\end{equation}
		where $L$ denotes the number of paths, $\theta_{\ell}\in [-\pi/2,\pi/2]$ is the angle-of-arrival (AoA) of the $\ell$-th signal path, $\tau_{\ell} \in [0,\tau_{\max}]$ is the path delay where $\tau_{\max}$ is the
		maximum delay spread of the channel which is bounded by the length of the OFDM cyclic prefix, $c_\ell\in \bC$ is the complex gain of path $\ell$ and $\Delta f$ is the subcarrier spacing (assumed uniform), $\lambda$ is the wavelength corresponding to the central carrier frequency, and $d$ is the uniform spacing between array elements, taken to be $d=\lambda/2$ for simplicity. The path AoAs and delays are chosen uniformly at random and gains are generated as standard Gaussian random variables. The resulting channel covariance is normalized such that ${\rm Tr}(\Sigmam_h) = MN$. The number of paths $L$ is the same for all UEs, but the AoAs, delays and gains are generated independently across UEs. Because the AoAs and delays are randomly generated, the channel covariance rank is given, with probability one, by $r=\min(L,MN )$.

	We first study the CSIT estimation MSE performance as a function of SNR for all the feedback strategies. We generate 10 realizations of training matrices and $K$ random covariances as above, and for each realization, we generate 100 random instances of each UE's channel. Here the covariance rank is set to $r=30$. The \textit{average} MSE is computed as
	\[\text{MSE}_{\text{avg}} = \frac{1}{K}\sum_k E[\Vert \hv_k- \widehat{\hv}_k  \Vert^2],\]
	where the mean $E[\Vert \hv_k -\widehat{\hv}_k  \Vert^2]$ is empirically calculated from the random realizations of training matrices, covariances and channels. We plot the the average MSE against DL SNR for two points in the $\beta_{\rm tr}-\beta_{\rm fb}$ plane, namely $(\beta_{\rm tr},\beta_{\rm fb}) = (40,40)$ and $(40,10)$. From \eqref{eq:decay_exponent_ub} and \eqref{eq:af_decay_exponent} we expect that all feedback methods achieve a quality scaling exponent of $\alpha = 1$ since $\beta_{\rm tr} = \beta_{\rm fb}> r$. In the second case however, we expect that the rate-distortion and ECSQ achieve a quality scaling exponent of $\alpha = \beta_{\rm fb}/r = 1/3$, while AF achieves $\alpha = 0$ because $\beta_{\rm tr}>r>\beta_{\rm fb}$. These are confirmed by the average MSE vs SNR curves of Fig. \ref{fig:mse_vs_snr}, where the slope of the curves in high SNR is equivalent to the quality scaling exponent. In the first case, AF achieves the optimal quality scaling exponent because both training and feedback dimensions exceed the covariance rank. In the second case, the rate-distortion and ECSQ feedbacks achieve a non-zero exponent, unlike AF which has a constant error ($\alpha=0$) even for large SNR. Note also that ECSQ yields an error close to the optimal, with a performance gap between the two that widens when the feedback dimension is lower. This is expected, since when the feedback rate is higher, the gap between the error achieved by scalar quantization and the optimal error is lower.

	

	In the second experiment, we compare the performance of different feedback strategies in terms of achievable Downlink sum-rate in a multi-user system. Here we consider ZF precoding in DL, where the transmit data vector over subcarrier $n$ is given by $\xv^{\rm d}[n] = \sum_{k=1}^K \sqrt{P_k} s_k [n] \vv_k [n]$, where $P_k = \snrdl /K$ is the (uniform) transmission power per UE, $s_k [n]\in \bC$ is the data symbol intended for UE $k$ such that $\bE[|s_k [n]|^2]\le 1$, and $\{\vv_k [n]\}_{k=1}^K$ are the precoding vectors, given by the column-normalized pseudo-inverse of \scalebox{0.9}{$\widetilde{\Hm}[n]=[\widetilde{\hv}_1[n],\ldots,\widetilde{\hv}_K[n]]^\herm$}. The choice of ZF precoding (rather than rate-splitting) is for the sake of simplicity and the fact that it is by far the most practical scheme used in 
	real systems. Defining the variables $g_{k,k'}[n] = \sqrt{P_{k'}} (\widetilde{\hv}_{k}[n]^\herm \vv_{k'}[n])$, we can write the achievable ergodic rate\footnote{This rate can be achieved assuming perfect knowledge of the coefficients $\{g_{k,k'}[n] \}$ at UE $k$, and is a harmless assumption for our purposes.} as \cite{caire2018ergodic}
	\begin{equation}\label{eq:ach_rate}
		R_{\rm ub} [k,n] =\bE \left[\log \left(1+ \frac{|g_{k,k}[n]|^2}{N_0 + \sum_{k'\neq k} |g_{k,k'}[n]|^2}\right) \right].
	\end{equation}
	\begin{figure}[t]
	\centerline{\includegraphics[scale=0.75]{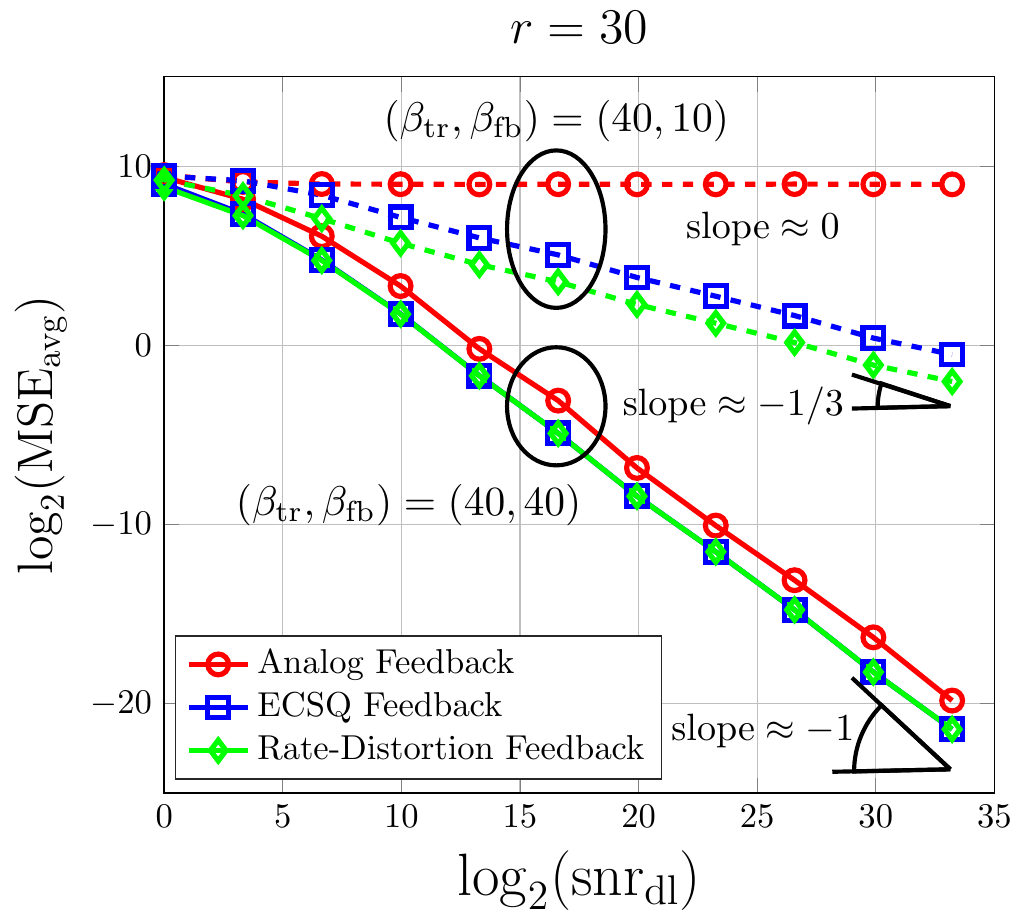}
	}
	\caption{CSIT estimation MSE vs Downlink SNR for the three feedback strategies, where $M=32$, $N=24$ and $K=6$.}
	\label{fig:mse_vs_snr}
\end{figure}
The average sum-rate is computed by averaging \eqref{eq:ach_rate} over all subcarriers, distinguishing between pilot and data subcarriers, and summing the result over all UEs. Fig. \ref{fig:rate_vs_beta} illustrates two sets of curves, comparing the sum-rate vs training dimension for the three feedback strategies. In the first set, we have assumed a high DL SNR value of $50$ dBs and we have set $\zeta = \beta_{\rm fb}/\beta_{\rm tr}=1$, so that at each point of the associated curves, the number of training and feedback dimensions are equal. This corresponds to a line trajectory in the $\beta_{\rm tr}-\beta_{\rm fb}$ plane. From Fig. \ref{fig:quality_scaling_exponent} we expect that, moving along the line $\beta_{\rm fb} = \beta_{\rm tr}$, AF achieves the same quality scaling exponent as the optimal rate-distortion feedback, which implies that in high SNR the two feedback strategies should have close rate performance. This is confirmed by the first set of curves in Fig. \ref{fig:rate_vs_beta} where we see that AF achieves a sum-rate that is very close to that of the rate-distortion feedback.
	
	The second set of curves corresponds to a moderate DL SNR of $20$ dBs, and $\zeta = 1/4$, which means that for each point of the curves the feedback dimension is taken to be $1/4$ the training dimension (rounded up when $\beta_{\rm tr}/4$ is not an integer). In this case, we have a noticeable gap between the sum-rate with AF and that of the optimal because the SNR is set to a moderate value and more importantly, because of the fact that in this case for training dimension values of $\beta_{\rm tr} \in(30,120)$, we have $\beta_{\rm fb} = \beta_{\rm tr}/4 \in (7,30)$, which means that for these points we have $(\beta_{\rm tr},\beta_{\rm fb}) \in \Rc_2$ (see Fig. \ref{fig:quality_scaling_exponent}). In the region $\Rc_2$, AF is \textit{strictly} suboptimal, in the sense that while the optimal quality scaling exponent is $\alpha_{\rm rd} = \beta_{\rm fb}/r \in (1/4,1)$, for AF it is $\alpha_{\rm af} = 0$. 
	\begin{figure}[t]
		\centerline{\includegraphics[scale=0.75]{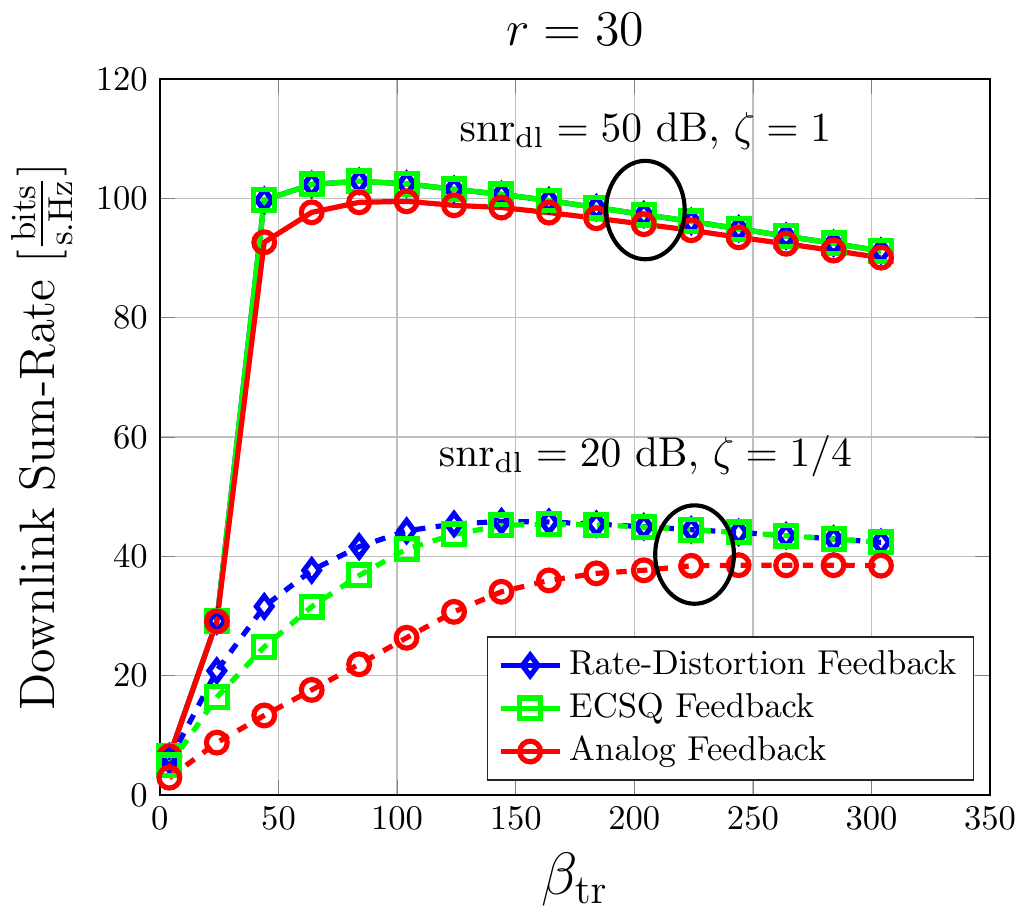}
		}
		\caption{Achievable sum-rate vs training dimension for the three feedback strategies, where $M=32$, $N=24$ and $K=6$.}
		\label{fig:rate_vs_beta}
	\end{figure}
	
	Note also the gap between the sum-rate with ECSQ and the optimal in the two sets of curves in Fig. \ref{fig:rate_vs_beta}. In the first set, because we have a large SNR and a large feedback rate ($\zeta=1$ in comparison to $\zeta=1/4$), the CSIT error achieved by ECSQ for the same pair of training and feedback dimensions is very close to the optimal. However, in the second case we have a relative shortage of feedback bits, so that the 1.5 bit per quantized coefficient overhead of ECSQ corresponds to a noticeable CSIT error and we see a larger gap between the sum-rates. Nevertheless, ECSQ is still close to the optimal and performs better than AF, which comes at the price of channel covariance knowledge at the UE side. Therefore there exists a decision point where we can choose the ECSQ feedback in case the channel covariance is available at the UE and the AF in case it is not.

	\appendix
	\subsection{Proof of Lemma \ref{lem:rate_distortion}}\label{app:RD_lemma_proof}
	We start by stating a few standard results regarding the (remote) rate-distortion function. It is well-known that the rate-distortion function of an i.i.d source represented by the random variable $\uv$ with distribution $p_{\uv}$ can be computed as (see \cite{thomas2006elements} Theorem 10.2.1)
	\begin{equation}\label{eq:rate_distortion}
		R_{\uv} (D) = \min_{p_{\widebar{\uv}|\uv}: d(\widebar{\uv},\uv)\le D} I(\widebar{\uv};\uv),
	\end{equation}
	where $\widebar{\uv}$ is the quantization of $\uv$, $I(\widebar{\uv};\uv)$ is the mutual information between $\uv$ and $\widebar{\uv}$, $d(\widebar{\uv},\uv)$ is the distortion between $\uv$ and $\widebar{\uv}$ (see \eqref{eq:err_def}) and the minimum is taken over all conditional distributions for which the joint distribution $p_{\uv,\widebar{\uv}} $ satisfies the distortion constraint. It is also known that the remote rate distortion function of a source represented by the random variable $\hv$, and encoded given its observations denoted by the random variable $\yv^{\rm tr}$ is given by (see \cite{berger1971rate})
	\begin{equation}\label{eq:remote_rate_distortion}
		R_{\hv}^r (D) = \min_{p_{\widebar{\hb}|\yb^{\rm tr}} : d(\widebar{\hv},\hv)\le D}~I(\widebar{\hv};\yv^{\rm tr})
	\end{equation}
	where $\widebar{\hv}$ is the quantization of $\hv$, $I(\widebar{\hv};\yv^{\rm tr})$ is the mutual information between $\widebar{\hv}$ and $\yv^{\rm tr}$, and the minimum is taken over all conditional distributions $p_{\widebar{\hv}|\yv^{\rm tr}}$ for which the joint distribution $p_{\hv,\widebar{\hv}}= p_{\hv} p_{\yv^{\rm tr}|\hv}p_{\widebar{\hv}|\yv^{\rm tr}} $ satisfies the distortion constraint. Note that since all sources are i.i.d, we have removed realization index superscripts from the variables (hence $\hv$ instead of $\hv^{(i)}$). From the premise of the lemma, $\uv$ is the MMSE estimate of the channel given the training measurements, i.e. $\uv = \bE [\hv|\yv^{\rm tr}]$. Using the same technique employed to prove inequality (15) of \cite{eswaran2019remote} (see Appendix A in \cite{eswaran2019remote}), we can verify that the remote rate-distortion function of $\hv$ is related to the rate distortion function of $\uv$ by
	\begin{equation}\label{eq:relation_1}
		R_{\hv}^r (D)  = R_{\uv}(D-D_{\rm mmse}),
	\end{equation}
	for $D\ge D_{\rm mmse}$, where $D_{\rm mmse} = \bE \left[\Vert \hv - \uv\Vert^2 \right]$ is the MMSE of estimating the channel at the UE. 
	
	On the other hand, the rate-distortion function of a correlated vector Gaussian source is given by reverse water-filling over its covariance eigenvalues \cite{cover2006elements}. If we denote the eigenvalues of $\Sigmam^u$ by $\{ \lambda_{\ell}^u \}_{\ell=1}^{MN}$, then we have $R_{\uv} (D) = \sum_{\ell=1}^{MN}\left[ \log \frac{\lambda_{\ell}^u}{\gamma} \right]_+,$ where $\gamma $ is chosen such that $\sum_{\ell=1}^{MN}\min \{ \gamma, \lambda_{\ell}^u \}=D$. Plugging this in \eqref{eq:relation_1} we get
	\begin{equation}
		R_{\hv}^r (D) = \sum_{\ell=1}^{MN}\left[ \log \frac{\lambda_{ \ell}^u}{\gamma} \right]_+,
	\end{equation}
	where $\gamma$ is chosen such that $\sum_{\ell=1}^{MN}\min \{ \gamma, \lambda_{u, \ell} \} = D-D_{\rm mmse}.$ The proof is complete. \hfill $\blacksquare$

	\subsection{Proof of Theorem \ref{thm:rate_asymp}}\label{app:rate_asymp_thm_proof}
	We divide the proof to two parts. In the first part, we show that if $\beta_{\rm tr} < r$, then the achievable error behaves as $\Theta (1)$ for all realizations of $\Xm^{\rm tr}$. In the second part we show that if $\beta_{\rm tr} \ge r$, then an error decaying as $\Theta (\snrdl^{-\min(\beta_{\rm fb}/r,1)})$ is achievable with probability one over the realizations of $\Xm^{\rm tr}$.\\
	
	\noindent\textbf{Part I.} To prove part I, we first bound the minimum mean squared error (MMSE) of estimating the channel given the training measurements at the UE, namely the variable $D_{\rm mmse}$. From $\uv = \bE [\hv | \yv^{\rm tr}]$ we have
	\begin{equation}\label{eq:Dmmse}
		D_{\rm mmse} = \trace \left(\bE [\hv\hv^\herm ] - \bE[\hv \yv^{\rm tr} ]\bE[\yv^{{\rm tr}\, \herm} \yv^{\rm tr} ]^{-1} \bE[\hv \yv^{\rm tr} ]^\herm  \right),
	\end{equation}
	where $\bE [\hv\hv^\herm ] = \Sigmam^h$, $\bE[\hv \yv^{\rm tr} ] =\Sigmam^h \Xm^{\rm tr} $, and $\bE[\yv^{{\rm tr}\, \herm} \yv^{\rm tr} ] = \Xm^{{\rm tr}\, \herm }\Sigmam^h \Xm^{\rm tr} + \mathbf{I} $. The eigendecomposition of $\Sigmam^h$ can be written as $\Sigmam^h = \Um_h \Lambdam_h \Um_h^\herm$, where $\Um_h\in \bC^{MN\times r} $ is a tall unitary matrix and $\Lambdam_h = {\rm diag}(\lambdav) \in \bR^{r\times r}$ is a diagonal matrix of positive eigenvalues represented by the vector $\lambdav=[\lambda_1,\ldots,\lambda_r]^\transp$. Using this decomposition and applying the Sherman-Morrison-Woodbury matrix identity to $\left( \Xm^{{\rm tr}\, \herm }\Sigmam^h \Xm^{\rm tr} + \mathbf{I} \right)^{-1}$, we have
	\begin{equation}
		\begin{aligned}
			\bE[\hv \yv^{\rm tr} ]&\bE[\yv^{{\rm tr}\, \herm} \yv^{\rm tr} ]^{-1} \bE[\hv \yv^{\rm tr} ]^\herm=\Um_h \Lambdam_h^{1/2}\Gm \Lambdam_h^{1/2}\Um_h^\herm \\&- \Um_h\Lambdam_h^{1/2}\Gm \left( \mathbf{I} + \Gm \right)^{-1}\Gm\Lambdam_h^{1/2}\Um_h^\herm,
		\end{aligned}
	\end{equation}
	where we have defined 
	\begin{equation}\label{eq:G_def_0}
		\begin{aligned}
			\Gm = \Lambdam_h^{1/2}  \Um_h^\herm \Xm^{\rm tr} \Xm^{{\rm tr}\, \herm }\Um_h  \Lambdam_h^{1/2}. 
		\end{aligned}
	\end{equation}
	Plugging this into \eqref{eq:Dmmse} we have
	\begin{equation}\label{eq:Dmmse_2}
		D_{\rm mmse} = \trace \left( \Lambdam_h \left(\mathbf{I} -  \Gm + \Gm ( \mathbf{I}+\Gm)^{-1} \Gm \right)\right). 
	\end{equation}
	Using a simple trace inequality, one can show that 
	\begin{equation}\label{eq:mmse_bound}
		\lambda_{\rm min} \, g(\snrdl) \le D_{\rm mmse}\le \lambda_{\rm max} \, g(\snrdl),
	\end{equation}
	where $\lambda_{\rm min}$ and $\lambda_{\rm max}$ are minimum and maximum channel covariance eigenvalues, respectively, and we have defined  $g(\snrdl) = \trace \left( \mathbf{I} -  \Gm + \Gm ( \mathbf{I}+\Gm)^{-1} \Gm \right) $, where we have made the dependency of $g(\cdot)$ on $\snrdl$ explicit. We now demonstrate this dependency.
	For a given realization of the training matrix $\Xm^{\rm tr}$, denote the eigenvalues of $\Gm$ by $\mu_{i},\, i=1,\ldots, r$. We can write
	\begin{equation}\label{eq:tr_expansion_0}
		\begin{aligned}
			g(\snrdl) &= r - \sum_i \mu_i +\sum_i \frac{\mu_i^2 }{\mu_i + 1} \\ 
			& =  r - \sum_{i=1}^r \frac{\mu_i }{\mu_i + 1}
		\end{aligned}
	\end{equation} 
	 Also note that we can represent the training matrix as $\Xm^{\rm tr} = \sqrt{\snrdl}\Xm_0^{\rm tr}$, where $\Xm_0^{\rm tr}$ is randomly generated and independent from $\snrdl$. From this and the definition \eqref{eq:G_def_0}, it follows that  $\mu_i = \Theta (\snrdl)$ for all $\mu_i \neq 0$. Using this and \eqref{eq:tr_expansion_0}, we deduce that if $\Gm$ is full-rank ($\mu_i\neq 0$ for all $i$) then $g(\snrdl) = \Theta (\snrdl^{-1})$ and using \eqref{eq:mmse_bound}  we have $ D_{\rm mmse} = \Theta (\snrdl^{-1})$. Conversely, if $\Gm$ has at least one zero eigenvalue ($ \mu_i = 0$ for some $i$) then from \eqref{eq:tr_expansion_0} we have $ g(\snrdl) >1$ and from \eqref{eq:tr_bounds} we have $ D_{\rm mmse} = \Theta(1)$.
	
	Now, the rank of $\Gm$ depends on the specific realization of $\Xm^{\rm tr}$. When $\beta_{\rm tr}\ge r$ and $\Xm^{\rm tr}$ consists of Gaussian isotropic pilot vectors, $\Gm$ is full-rank with probability one because of the following. The product $\Um_h^\herm \Xm^{{\rm tr}}$ consists of $\beta_{\rm tr} $ independent Gaussian columns, each of dimension $r$. The event that these vectors span a space of dimension less than $r$ has probability zero. Therefore, $\Um_h^\herm \Xm^{\rm tr}  \Xm^{{\rm tr}\, \herm }\Um_h$ has rank $r$ with probability one, and since $\Lambdam_h^{1/2}$ has positive diagonal elements, by definition \eqref{eq:G_def_0} $\Gm$ also has rank $r$ with probability one and $\mu_i \neq 0$ for all $i$. Conversely, if $\beta_{\rm tr}< r$, $\Gm$ has rank at most equal to $\beta_{\rm tr}$ for any realization of the training matrix, leading to $\mu_i= 0$ for some $i$. This results in $D_{\rm mmse}=\Theta (1)$. In short, we have proved
	\begin{equation}
		D_{\rm mmse} = \begin{cases}
			\Theta (\snrdl^{-1}),  &\beta_{\rm tr} \ge r,\\
			\Theta (1),&  \beta_{\rm tr} < r.
		\end{cases}
	\end{equation}
	In addition, Lemma \ref{lem:rate_distortion} states that only errors $D\ge D_{\rm mmse}$ are achievable. It follows that, if $\beta_{\rm tr}< r$, then the minimum achievable error in estimating the CSIT behaves as $\Theta (1)$. \\

	\noindent\textbf{Part II.} To prove the second part, first note that if $\beta_{\rm tr} \ge r$, then the covariance of the MMSE channel estimate $\uv$ at the UE, given as
	\begin{equation}\label{eq:Sigmau}
	\Sigmam^u = \Sigmam^h \Xm^{\rm tr}\left(  \Xm^{{\rm tr}\, \herm} \Sigmam^h \Xm^{{\rm tr}} + \mathbf{I} \right)^{-1}\Xm^{{\rm tr}\, \herm} \Sigmam^h
	\end{equation}
	 has rank $r$ with probability one over the realizations of $\Xm^{\rm tr}$. Without loss of generality assume the eigenvalues of $\Sigmam^u$ to be ordered as $\lambda_{1}^u \ge \ldots \ge \lambda_{r}^u > 0  $. Next, consider the remote rate-distortion function in Lemma \ref{lem:rate_distortion}, given as
	\begin{equation}
		R_{\hv}^r (D) = \sum_{\ell=1}^{MN}\left[ \log \frac{\lambda_{\ell}^{u}}{\gamma} \right]_+,~\text{for}~ D\ge D_{\rm mmse}
	\end{equation}
	where $\gamma$ is chosen such that $\sum_{\ell=1}^{MN}\min \{ \gamma, \lambda_{\ell}^{u} \} = D-D_{\rm mmse}.$ Consider an interval of error values $D$ for which $D - D_{\rm mmse}<\varepsilon$ for some $\varepsilon$. If $\varepsilon$ is sufficiently small, then $\gamma = (D-D_{\rm mmse})/r$ and the remote rate-distortion function is given by 
	\begin{equation}\label{eq:R_formula}
		 R_{\hv}^r (D) =f(r) - r \log (D-D_{\rm mmse}),
	\end{equation}
	where $f(r)= \sum_{\ell = 1}^{r} \log \lambda_{ \ell}^u  + r\log r$ is a value independent of $D$. Therefore, we can write the remote distortion-rate function as 
	\begin{equation}\label{eq:dist_rate_func}
		D_{\hv}^r (R) = 2^{\frac{f(r)-R}{r}} + D_{\rm mmse}.
	\end{equation}
	for all $R>R_{\varepsilon}$, for a sufficiently large $R_{\varepsilon}>0$. Now let $R=\beta_{\rm fb} C^{\rm ul}$. Replacing the MIMO-MAC capacity formula $C^{\rm ul} = \log (1+ M\kappa \snrdl)$, we notice that there exists some $\snrdl^\varepsilon$ such that $\beta_{\rm fb} \log (1+M\kappa \snrdl)>R_{\varepsilon}$ for all $\snrdl>\snrdl^{\varepsilon}$. For these $\snrdl$ values we have we have
	\begin{equation}\label{eq:my_eqqq}
		\begin{aligned}
		\log ( D_{\hv}^r (\beta_{\rm fb} C^{\rm ul})- D_{\rm mmse}) &= \log r + \sum_{\ell=1}^r \log \lambda_{ \ell}^u /r\\
		&- \frac{\beta_{\rm fb}}{r}\log (1+ M\kappa \snrdl).
	\end{aligned}
	\end{equation}
	From \eqref{eq:Sigmau} one can show that the non-zero eigenvalues of $\Sigmam^u$ scale as $\Theta (1)$ for large $\snrdl$, i.e. $\lambda_{\ell}^u = \Theta(1),\, \ell=1,\ldots,r$. Therefore, the right-hand-side of \eqref{eq:my_eqqq} behaves as $\Theta(\log (\snrdl^{-\beta_{\rm fb}/r}))$ in $\snrdl$. It follows that 
	\begin{equation}\label{eq:my_eq_2}
		\begin{aligned}
			 D_{\hv}^r (\beta_{\rm fb} C^{\rm ul}) &= D_{\rm mmse} +\Theta(\snrdl^{-\beta_{\rm fb}/r})\\
			 & =\Theta(\snrdl^{-1})+\Theta(\snrdl^{-\beta_{\rm fb}/r})\\
			 & =\Theta(\snrdl^{-\min(\beta_{\rm fb}/r,1)}) 
		\end{aligned}
	\end{equation}
	Finally, from the source-channel separation with distortion theorem, we can achieve a CSIT estimation error of $D > D_{\hv}^r ( \beta_{\rm fb} C^{\rm ul} )$ if and only if we use the UL channel over $\beta_{\rm fb}$ feedback dimensions (see Section \ref{sec:RD_LB}), which combined with \eqref{eq:my_eq_2} shows that when $\beta_{\rm tr}\ge r$, we can achieve an error decay of $\Theta(\snrdl^{-\min(\beta_{\rm fb}/r,1)}) $ with probability one over the realizations of $\Xm^{\rm tr}$ with a feedback dimension of $\beta_{\rm fb}$.
	This completes the proof.\hfill $\blacksquare$
		\subsection{Proof of Theorem \ref{thm:af_distortion_modified}}\label{app:af_thm_proof_modified}
		The distortion of MMSE estimation of the channel given the feedback in \eqref{eq:af_bs_signal} is given by 
		\begin{equation}\label{eq:af_distortion_modified}
			d(\hv,\widehat{\hv}) = \trace \left ( \bE[\hv \hv^\herm] - \bE[\hv \yv^{{\rm af}}] \bE[\yv^{\rm af\, \herm} \yv^{{\rm af}}]^{-1} \bE[\hv \yv^{{\rm af}}]^\herm\right) , 
		\end{equation}
		where $\Sigmam^h=\bE[\hv \hv^\herm] $,
		$\bE[\hv \yv^{{\rm af}}] =  \Sigmam^{h}\Xm^{\rm tr}\Psim ,$ and 
		\begin{equation}\label{eq:y_af_cov}
			\bE[\yv^{\rm af\, \herm} \yv^{{\rm af}}]= \Psim^\herm \Xm^{{\rm tr}\, \herm} \Sigmam^h \Xm^{{\rm tr}}\Psim +\Psim^\herm \Psim +\mathbf{I},
		\end{equation}
		The eigendecomposition of $\Sigmam^h$ can be written as $\Sigmam^h = \Um_h \Lambdam_h \Um_h^\herm$, where $\Um_h\in \bC^{MN\times r} $ is a tall unitary matrix and $\Lambdam_h = {\rm diag}(\lambdav) \in \bR^{r\times r}$ is a diagonal matrix of positive eigenvalues represented by the vector $\lambdav=[\lambda_1,\ldots,\lambda_r]^\transp$. Using this decomposition, the expression in \eqref{eq:y_af_cov} and applying the Sherman-Morrison-Woodbury matrix identity we have
		\begin{equation}
			\begin{aligned}
				\bE [\yv^{{\rm af}\, \herm} \yv^{\rm af} ]^{-1} &= (\Psim^\herm \Psim + \mathbf{I})^{-1} \\
				& -(\Psim^\herm \Psim + \mathbf{I})^{-1} \Psim^\herm \Xm^{{\rm tr}\, \herm}\Um_h \Lambdam_h^{1/2}\times \\
				&\hspace{-20mm} \left( \mathbf{I} + \Lambdam_h^{1/2}\Um_h^\herm \Xm^{{\rm tr}}\Psim (\Psim^\herm \Psim + \mathbf{I})^{-1} \Psim^\herm \Xm^{{\rm tr}\, \herm }\Um_h \Lambdam_h^{1/2}    \right)^{-1} \times\\
				& \Lambdam_h^{1/2} \Um_h^\herm \Xm^{\rm tr} \Psim (\Psim^\herm \Psim +\mathbf{I})^{-1}.
			\end{aligned}
		\end{equation}
		Then we can write the second term appearing within the ${\rm Tr}(\cdot)$ in \eqref{eq:af_distortion_modified} as
		\begin{equation}\label{eq:eqqq0}
			\begin{aligned}
				\bE[\hv \yv^{{\rm af}}] \bE[\yv^{\rm af\, \herm} \yv^{{\rm af}}]^{-1} \bE[\hv \yv^{{\rm af}}]^\herm &=\Um_h \Lambdam_h^{1/2}\Gm \Lambdam_h^{1/2}\Um_h^\herm \\&\hspace{-20mm}- \Um_h\Lambdam_h^{1/2}\Gm \left( \mathbf{I} + \Gm \right)^{-1}\Gm\Lambdam_h^{1/2}\Um_h^\herm,
			\end{aligned}
		\end{equation}
		where we have defined 
		\begin{equation}\label{eq:G_def}
			\begin{aligned}
				\Gm &= \Lambdam_h^{1/2}\Um_h^\herm \Xm^{\rm tr}\Psim (\Psim^\herm \Psim + \mathbf{I})^{-1} \Psim^\herm \Xm^{{\rm tr}\, \herm} \Um_h\Lambdam_h^{1/2}\\
			\end{aligned}
		\end{equation}
		Plugging \eqref{eq:eqqq0} into \eqref{eq:af_distortion_modified} we have
		\begin{equation}\label{eq:distor_1}
			d(\hv,\widehat{\hv}) = \trace \left( \Lambdam_h \left(\mathbf{I} -  \Gm + \Gm ( \mathbf{I}+\Gm)^{-1} \Gm \right)\right). 
		\end{equation}
		This formula is exactly the same as \eqref{eq:Dmmse_2} except for the definition of $\Gm$. Therefore the same trace inequality as in \eqref{eq:mmse_bound} holds here for the CSI estimation distortion at the BS, i.e. we have 
		\begin{equation}\label{eq:tr_bounds}
			\lambda_{\rm min} \, g(\snrdl) \le d(\hv,\widehat{\hv}) \le \lambda_{\rm max} \, g(\snrdl),
		\end{equation}
		where $g(\snrdl) = \trace \left( \mathbf{I} -  \Gm + \Gm ( \mathbf{I}+\Gm)^{-1} \Gm \right) $. We now show how $g(\cdot)$ behaves as a function of $\snrdl$. For a given realization of $\Xm^{\rm tr}$, denote the eigenvalues of $\Gm$ by $\mu_{i},\, i=1,\ldots, r$. We can write
		\begin{equation}\label{eq:tr_expansion}
			\begin{aligned}
				g(\snrdl) &= r - \sum_i \mu_i +\sum_i \frac{\mu_i^2 }{\mu_i + 1} \\ 
				& =  r - \sum_i \frac{\mu_i }{\mu_i + 1}
			\end{aligned}
		\end{equation} 
		From the definition in \eqref{eq:G_def}, the constituents of $\Gm$ depend on $\snrdl$ as follows:
		\begin{itemize}
			\item[(a)] We can represent the training matrix as $\Xm^{\rm tr} = \sqrt{\snrdl}\Xm_0^{\rm tr}$, where $\Xm_0^{\rm tr}$ is generated randomly independent from $\snrdl$. Hence, for a single realization the elements of $\Xm^{\rm tr}$ scale with $\snrdl$ as $\Theta (\sqrt{\snrdl})$.
			\item[(b)] From constraint \eqref{eq:psi_vecs}, one can show that each column of $\Psim$ can be written as $\psiv_i = \sqrt{a_i} \phiv_i$, where $a_i = \frac{M\kappa \snrdl }{c_i \snrdl +1}$ with $c_i = \phiv_i^\herm \Xm_0^{{\rm tr}\, \herm }\Sigmam_h \Xm^{\rm tr}_0 \phiv_i$. Here $\phiv_i, \, i\in [\beta_{\rm fb}]$ are a set of unit-norm vectors that contains a subset of $\min(\beta_{\rm tr},\beta_{\rm fb})$ linearly independent vectors, and are independent of $\snrdl$. The existence of this set is guaranteed because $\Sigmam_{\yv^{\rm tr}}$ is full-rank. It follows that
			\begin{equation}
				\begin{aligned}
					\Psim (\Psim^\herm \Psim + \mathbf{I} )^{-1} \Psim^\herm = \Phim (\Phim^\herm \Phim + \Sm )^{-1} \Phim^\herm, 
				\end{aligned}
			\end{equation}
			where $\Phim = [\phiv_1,\ldots,\phiv_{\beta_{\rm fb}}]$ is independent of $\snrdl$ and $\Sm$ is a diagonal matrix whose $(i,i)$ element is given by $S_{i,i} = \frac{c_i}{M\kappa}+ \frac{1}{M\kappa \snrdl}$. The matrix $\Sm$ is the only variable dependent on $\snrdl$, and its diagonal elements are asymptotically scaling as $S = \Theta (1)$. 
			\item[(c)] The matrices $\Um_h$ and $\Lambdam_h$ are independent of $\snrdl$.
		\end{itemize}
		From these we conclude that the non-zero eigenvalues of $\Gm$ scale as $\Theta (\snrdl )$, i.e. $\mu_i = \Theta (\snrdl)$ for all $\mu_i \neq 0$.

		Now, the rank of $\Gm$ depends on the specific realization of $\Xm^{\rm tr}$. When $\min(\beta_{\rm tr},\beta_{\rm fb})\ge r$ and $\Xm^{\rm tr}$ contains isotropic Gaussian pilot vectors, $\Gm$ is full-rank with probability one because of the same argument as used in Part I of the proof of Theorem \ref{thm:rate_asymp} and considering the fact that the constituent matrix $\Psim (\Psim^\herm \Psim + \mathbf{I} )^{-1} \Psim^\herm$ is positive semi-definite with rank $\min (\beta_{\rm tr},\beta_{\rm fb})$. In this case we have $d (\hv,\widehat{\hv})=\Theta (\snrdl^{-1})$ and therefore an error of $\Theta (\snrdl^{-1})$ is achievable. Conversely, if $r>\beta_{\rm tr}$ or $r>\beta_{\rm fb}$, $\Gm$ has rank at most $\min (\beta_{\rm tr},\beta_{\rm fb}) < r$ for any design of pilot matrices, leading to $\mu_i= 0$ for some $i$ and from \eqref{eq:tr_bounds}, the error is bounded from below and above by constants, i.e. we have an error of $\Theta (1)$. Therefore AF achieves an error of $\Theta(\snrdl^{-\alpha}) $ with probability one over the realizations of $\Xm^{\rm tr}$ where $\alpha = \mathbf{1}_{ [r,\infty) }\left(\min(\beta_{\rm tr},\beta_{\rm fb}) \right)$. This completes the proof.\hfill $\blacksquare$
	
	{\small
		\bibliographystyle{IEEEtran}
		\bibliography{references}

\begin{thebibliography}{10}
\providecommand{\url}[1]{#1}
\csname url@samestyle\endcsname
\providecommand{\newblock}{\relax}
\providecommand{\bibinfo}[2]{#2}
\providecommand{\BIBentrySTDinterwordspacing}{\spaceskip=0pt\relax}
\providecommand{\BIBentryALTinterwordstretchfactor}{4}
\providecommand{\BIBentryALTinterwordspacing}{\spaceskip=\fontdimen2\font plus
\BIBentryALTinterwordstretchfactor\fontdimen3\font minus
  \fontdimen4\font\relax}
\providecommand{\BIBforeignlanguage}[2]{{%
\expandafter\ifx\csname l@#1\endcsname\relax
\typeout{** WARNING: IEEEtran.bst: No hyphenation pattern has been}%
\typeout{** loaded for the language `#1'. Using the pattern for}%
\typeout{** the default language instead.}%
\else
\language=\csname l@#1\endcsname
\fi
#2}}
\providecommand{\BIBdecl}{\relax}
\BIBdecl

\bibitem{larsson2014massive}
E.~G. Larsson, O.~Edfors, F.~Tufvesson, and T.~L. Marzetta, ``Massive {MIMO}
  for next generation wireless systems,'' \emph{IEEE communications magazine},
  vol.~52, no.~2, pp. 186--195, 2014.

\bibitem{boccardi2014five}
F.~Boccardi, R.~W. Heath, A.~Lozano, T.~L. Marzetta, and P.~Popovski, ``Five
  disruptive technology directions for {5G},'' \emph{IEEE communications
  magazine}, vol.~52, no.~2, pp. 74--80, 2014.

\bibitem{marzetta2006much}
T.~L. Marzetta, ``How much training is required for multiuser {MIMO}?'' in
  \emph{2006 fortieth asilomar conference on signals, systems and
  computers}.\hskip 1em plus 0.5em minus 0.4em\relax IEEE, 2006, pp. 359--363.

\bibitem{caire2007multiuser}
G.~Caire, N.~Jindal, M.~Kobayashi, and N.~Ravindran, ``Multiuser {MIMO}
  downlink made practical: Achievable rates with simple channel state
  estimation and feedback schemes,'' \emph{Arxiv preprint cs. IT}, vol. 710,
  2007.

\bibitem{jindal2006mimo}
N.~Jindal, ``{MIMO} broadcast channels with finite-rate feedback,'' \emph{IEEE
  Transactions on information theory}, vol.~52, no.~11, pp. 5045--5060, 2006.

\bibitem{biglieri1998fading}
E.~Biglieri, J.~Proakis, and S.~Shamai, ``Fading channels:
  Information-theoretic and communications aspects,'' \emph{IEEE transactions
  on information theory}, vol.~44, no.~6, pp. 2619--2692, 1998.

\bibitem{ziv1985universal}
J.~Ziv, ``On universal quantization,'' \emph{IEEE Transactions on Information
  Theory}, vol.~31, no.~3, pp. 344--347, 1985.

\bibitem{samardzija2006unquantized}
D.~Samardzija and N.~Mandayam, ``Unquantized and uncoded channel state
  information feedback in multiple-antenna multiuser systems,'' \emph{IEEE
  transactions on communications}, vol.~54, no.~7, pp. 1335--1345, 2006.

\bibitem{thomas2005obtaining}
T.~A. Thomas, K.~L. Baum, and P.~Sartori, ``Obtaining channel knowledge for
  closed-loop multi-stream broadband {MIMO-OFDM} communications using direct
  channel feedback,'' in \emph{GLOBECOM'05. IEEE Global Telecommunications
  Conference, 2005.}, vol.~6.\hskip 1em plus 0.5em minus 0.4em\relax IEEE,
  2005, pp. 5--pp.

\bibitem{marzetta2006fast}
T.~L. Marzetta and B.~M. Hochwald, ``Fast transfer of channel state information
  in wireless systems,'' \emph{IEEE Transactions on Signal Processing},
  vol.~54, no.~4, pp. 1268--1278, 2006.

\bibitem{kobayashi2011training}
M.~Kobayashi, N.~Jindal, and G.~Caire, ``Training and feedback optimization for
  multiuser {MIMO} downlink,'' \emph{IEEE Transactions on Communications},
  vol.~59, no.~8, pp. 2228--2240, 2011.

\bibitem{caire2007quantized}
G.~Caire, N.~Jindal, M.~Kobayashi, and N.~Ravindran, ``Quantized vs. analog
  feedback for the {MIMO} broadcast channel: A comparison between zero-forcing
  based achievable rates,'' in \emph{2007 IEEE International Symposium on
  Information Theory}.\hskip 1em plus 0.5em minus 0.4em\relax IEEE, 2007, pp.
  2046--2050.

\bibitem{jiang2015achievable}
Z.~Jiang, A.~F. Molisch, G.~Caire, and Z.~Niu, ``Achievable rates of {FDD}
  massive {MIMO} systems with spatial channel correlation,'' \emph{IEEE
  Transactions on Wireless Communications}, vol.~14, no.~5, pp. 2868--2882,
  2015.

\bibitem{shirani2009channel}
H.~Shirani-Mehr and G.~Caire, ``Channel state feedback schemes for multiuser
  {MIMO-OFDM} downlink,'' \emph{IEEE Transactions on Communications}, vol.~57,
  no.~9, pp. 2713--2723, 2009.

\bibitem{bazzi2018amount}
S.~Bazzi and W.~Xu, ``On the amount of downlink training in correlated massive
  {MIMO} channels,'' \emph{IEEE Transactions on Signal Processing}, vol.~66,
  no.~9, pp. 2286--2299, 2018.

\bibitem{gpp2020phylayer}
``Physical layer procedures for data,'' \emph{3GPP TS 38.214 version 16.2.0
  Release 16}, 2020.

\bibitem{kotecha2004transmit}
J.~H. Kotecha and A.~M. Sayeed, ``Transmit signal design for optimal estimation
  of correlated {MIMO} channels,'' \emph{IEEE Transactions on Signal
  Processing}, vol.~52, no.~2, pp. 546--557, 2004.

\bibitem{joudeh2016sum}
H.~Joudeh and B.~Clerckx, ``Sum-rate maximization for linearly precoded
  downlink multiuser {MISO} systems with partial {CSIT}: A rate-splitting
  approach,'' \emph{IEEE Transactions on Communications}, vol.~64, no.~11, pp.
  4847--4861, 2016.

\bibitem{davoodi2016aligned}
A.~G. Davoodi and S.~A. Jafar, ``Aligned image sets under channel uncertainty:
  Settling conjectures on the collapse of degrees of freedom under finite
  precision {CSIT},'' \emph{IEEE Transactions on Information Theory}, vol.~62,
  no.~10, pp. 5603--5618, 2016.

\bibitem{berger1971rate}
T.~Berger, ``Rate distortion theory, a mathematical basis for data compression
  (prentice-hall,'' \emph{Inc. Englewood Cliffs, New Jersey}, 1971.

\bibitem{cover2006elements}
T.~M. Cover and J.~A. Thomas, \emph{Elements of information theory}.\hskip 1em
  plus 0.5em minus 0.4em\relax John Wiley \& Sons, 2006.

\bibitem{sesia2011lte}
S.~Sesia, I.~Toufik, and M.~Baker, \emph{{LTE}-the {UMTS} long term evolution:
  from theory to practice}.\hskip 1em plus 0.5em minus 0.4em\relax John Wiley
  \& Sons, 2011.

\bibitem{sayeed2003virtual}
A.~M. Sayeed, ``A virtual representation for time-and frequency-selective
  correlated {MIMO} channels,'' in \emph{2003 IEEE International Conference on
  Acoustics, Speech, and Signal Processing, 2003. Proceedings.(ICASSP'03).},
  vol.~4.\hskip 1em plus 0.5em minus 0.4em\relax IEEE, 2003, pp. IV--648.

\bibitem{bajwa2008compressed}
W.~U. Bajwa, A.~Sayeed, and R.~Nowak, ``Compressed sensing of wireless channels
  in time, frequency, and space,'' in \emph{2008 42nd Asilomar Conference on
  Signals, Systems and Computers}.\hskip 1em plus 0.5em minus 0.4em\relax IEEE,
  2008, pp. 2048--2052.

\bibitem{caire2018ergodic}
G.~Caire, ``On the ergodic rate lower bounds with applications to massive
  {MIMO},'' \emph{IEEE Transactions on Wireless Communications}, vol.~17,
  no.~5, pp. 3258--3268, 2018.

\bibitem{thomas2006elements}
M.~Thomas and A.~T. Joy, ``Elements of information theory,'' 2006.

\bibitem{eswaran2019remote}
K.~Eswaran and M.~Gastpar, ``Remote source coding under {G}aussian noise:
  Dueling roles of power and entropy power,'' \emph{IEEE Transactions on
  Information Theory}, vol.~65, no.~7, pp. 4486--4498, 2019.

\end{thebibliography}
	}
	
\end{document}